\documentclass[conference,letterpaper]{IEEEtran}

\usepackage[noadjust]{cite}
\usepackage{amsmath,amssymb}
\usepackage{graphicx}
\usepackage{enumerate}
\usepackage{url} 
\usepackage{latexsym,epsfig,secdot,url,amsbsy,amsmath,exscale,color,fancyhdr,mathrsfs,enumerate,booktabs}
\usepackage{amsthm}    
\usepackage{amsfonts,euscript}
\usepackage{subcaption}
\usepackage[final]{hyperref} 
\hypersetup{
colorlinks=true,       
linkcolor=blue,        
citecolor=blue,        
filecolor=magenta,     
urlcolor=blue
}
\usepackage{soul}
\usepackage{bm}
\usepackage{algorithm,algpseudocode}
\usepackage{xcolor}
\newtheorem{theorem}{Theorem}
\newtheorem{lemma}{Lemma}

\newtheorem{corollary}{Corollary}
\newtheorem{asum}{Assumption}

\makeatletter
\newcommand*{\rom}[1]{\expandafter\@slowromancap\romannumeral #1@}
\makeatother

\begin{document}
\title{Quick and Consistent Sparsity Estimation for\\ Streaming Images with Noise}



\author{%
  \IEEEauthorblockN{Tingnan Gong}
  \IEEEauthorblockA{
                    School of Industrial and System Engineering\\
                    Georgia Institute of Technology\\
                    tgong33@gatech.edu}
}

\maketitle

\begin{abstract}
Given fruitful works in the image monitoring, there is a lack of data-driven tools guiding the practitioners to select proper monitoring procedures. The potential model mismatch caused by the arbitrary selection could deviate the empirical detection delay from their theoretical analysis and bias the prognosis. In the image monitoring, the sparsity of the underlying anomaly is one of the attributes on which the development of many monitoring procedures is highly based. This paper proposes a computational-friendly sparsity index,  the corrected Hoyer index, to estimate the sparsity of the underlying anomaly interrupted by noise. We theoretically prove the consistency of the constructed sparsity index. We use simulations to validate the consistency and demonstrate the robustness against the noise. We also provide the insights on how to guide the real applications with the proposed sparsity index. 
\end{abstract}
\section{Introduction} \label{sec:intro}
Among an abundance of monitoring methodologies \cite{qiu2018some}, the image monitoring is drawing attention since the development in sensor technology naturally allows the production line to real-time acquire high-resolution images in terms of high-dimensional matrices. Besides the traditional monitoring procedures, the diagnosis is emerging \cite{reis2017industrial}. As one aspect of the diagnosis, the estimation of the sparsity of the anomaly against the noise can reduce the model mismatch in terms of the deviated sparsity assumed in the applied monitoring procedure. 

The concept of sparsity is common in diverse areas involving image data, such as face recognition \cite{wright2008robust,wang2021deep}, image processing \cite{mairal2007sparse,aharon2008sparse,wright2010sparse,herath2017going} and medical imaging \cite{leung2008sparse, ferrante2017slice}. Also, among the advanced image monitoring procedures \cite{yan2018real,yan2022real,koosha2017statistical,Okhrin2021new,eslami2023spatial}, many of them are developed by assuming the anomalies to be sparse, which is seemingly natural for high-dimensional matrices. This is also known as defect detection since the defect in the production is often mere. The concern of model mismatch exists in that the monitoring procedures assumed to be applicable on sparse anomalies might suffer unexpected large delay on dense anomalies. For example, despite of the remarkable performance of the models in \cite{yan2018real,yan2022real} in detection of sparse anomaly, they allow dynamic background, which might absorb the dense anomaly into the background and delay the detection. 
In the community of computer vision, some literature study with dense anomalies in real applications \cite{sabokrou2015real,sabokrou2017deep}. 
\begin{figure}[ht!]
\centering 
\begin{subfigure}[h]{0.49\linewidth}
\includegraphics[width=\linewidth]{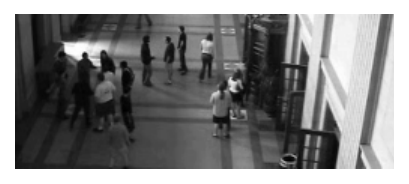}
\caption{The crowded in the hall.}
\end{subfigure}
\begin{subfigure}[h]{0.49\linewidth}
\includegraphics[width=\linewidth]{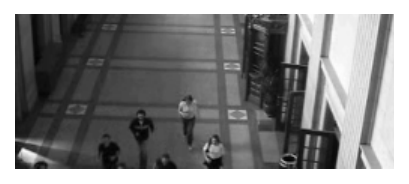}
\caption{The crowded fleeing.}
\end{subfigure}

\begin{subfigure}[h]{0.49\linewidth}
\includegraphics[width=\linewidth]{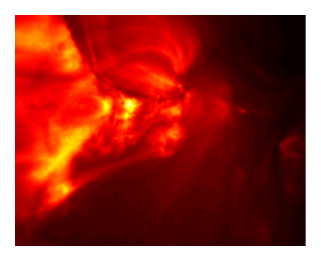}
\caption{Common activity.}
\end{subfigure}
\begin{subfigure}[h]{0.49\linewidth}
\includegraphics[width=\linewidth]{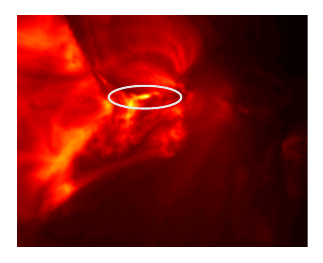}
\caption{Bright flashes.}
\end{subfigure}

\caption{Illustrative examples for dense and sparse anomalies. (a) Normal crowd activity. (b) The anomaly that the people flee from the scene. (c) Normal status of solar flare activity. (d) The solar flare outburst in the form of bright flashes.}
\label{fig:illustrative cases}
\end{figure}

\vspace{0pt}
As an illustrative example, the first row of Figure \ref{fig:illustrative cases} displays the normal pattern for indoor crowd activities on the left while in the right panel, the crowd flee from the scene and most of the pixels in the images have changed, namely a dense anomaly. In Figure \ref{fig:illustrative cases}, the second row displays the solar flare outburst phenomenon. Compared to the normal pattern in the left panel, the outburst in the right panel is visualized as a bright spot in the white circle. The anomaly in the solar flare activities happens within local area, namely a sparse anomaly. 

Our paper aims to consistently estimate the sparsity of the underlying anomaly to shed light on the legitimacy of the sparsity assumptions for the similar industry problems, conserve expensive computational resources and accelerate the future diagnosis and prognosis. The remainder of the paper is organized as follows. Section \ref{sec:problem} describes the problem setup and Section \ref{sec:method} proposes the scheme to compute the sparsity of the anomaly using Hoyer index. Section \ref{sec:theory} theoretically analyzes the bias of Hoyer index under the noise and constructs a corrected Hoyer index against the noise, whose consistency is proved. Section \ref{sec:simulation} designs simulated experiments to validate the robustness and consistency of the corrected Hoyer index. Section \ref{sec:realdata} conducts real data experiments to demonstrate that the corrected Hoyer index could estimate the sparsity of real images consistent with visual observations. 

\section{Problem setup}\label{sec:problem}
We denote the streaming images as  $\mathbf{X}_t\in\mathbb R^{{p_1}\times {p_2}}$ and model them in terms of matrices:
\begin{equation}\label{formula:model}
    \mathbf{X}_t = \bm{\mu}_0 + \mathbb{I}(t>\tau)\mathbf{A}+ \mathbf{e}_t,\quad t=1,2,\ldots,
\end{equation}
where $\bm{\mu}_0\in\mathbb R^{{p_1}\times {p_2}}$ is an in-control (IC) mean, $\mathbb{I}(B)$ is an indicator function equal to one if the event $B$ happens and otherwise zero, the integer $\tau$ is the unknown change point, the matrix $\mathbf A$ is the out-of-control (OOC) anomaly, and $\mathbf{e}_n$ is a noise matrix with randomness. If $\tau=\infty$, there is no change in the sequence of images. To accommodate general scenarios, we assume the entries of $\mathbf{e}_t$ to i.i.d. with zero-mean and stationary variances $\sigma^2$ without any specific constraints on their distributions and temporal correlations. Thus the proposed sparsity estimation method is distribution-free. For the temporal correlation, we have the assumption:
\begin{asum}\label{assump:IC-image-iid}
    $\{\mathbf{X}_t,\ t\le \tau\}$ are i.i.d..
\end{asum}
Assumption~\ref{assump:IC-image-iid} supports the convergence of the sample average of IC images to the true mean. This assumption can be relaxed in that the IC process is often sufficiently studied no matter the images are i.i.d. or not.  OOC images can have temporally varying distributions in $\mathbf{e}_t$ and unknown temporal correlation among $\mathbf{e}_t$. 

\section{Scheme to compute sparsity}\label{sec:method}
We introduce a quantitative index $h(\cdot):\mathbb R^{{p_1}\times {p_2}}\to [0,1]$ to measure the sparsity of the anomaly. \cite{hurley2009comparing} studied a variety of sparsity indices, among which the Hoyer index and Gini index are the best two, possessing the most desirable attributes out of six. Table~3 in \cite{hurley2009comparing} compares the sparsity indices. To be precise, the Gini index is the only one with all six attributes. However, the computation of the Gini index is relatively slow on high-dimensional data. The only failing attribute on Hoyer index is the invariance under data cloning, which is mild with high dimensions. Thus we choose Hoyer index to measure the sparsity of the anomaly. We define the Hoyer index on the individual image $\mathbf X\in\mathbb R^{{p_1}\times {p_2}}$ as:
\begin{equation}\label{formula:hoyer index}
    h(\mathbf X) = \left(\sqrt{p_1p_2} - \frac{\left|\sum_{i=1}^{p_1}\sum_{j=1}^{p_2} \mathbf X_{ij}\right|}{\|\mathbf X\|_F}\right) \left(\sqrt{p_1p_2}-1\right)^{-1}, 
\end{equation}
where $\|\cdot\|_F$ represents the Frobenius norm. To ensure that the Hoyer index in \eqref{formula:hoyer index} really reflect the sparsity of the anomaly $\mathbf A$, we have the assumption:
\begin{asum}\label{assump:non-nega-A}
    The entries of the anomaly $\mathbf A$, namely $\mathbf A_{ij}$, are same-sign. 
\end{asum}
Assumption \ref{assump:non-nega-A} is not artificial because many anomalies are one-sided, including the outburst of solar flare, the flee of the crowd, the missing spraying in the battery coating process and so on. If $h(\mathbf A) = 1$ holds, the anomaly is the most sparse with at most one non-zero entry. On the contrary, if $h(\mathbf A) = 0$ holds, the anomaly is the least sparse with all the entries to be the same. 
\subsection{Compute sparsity during monitoring}
Suppose a randomly selected monitoring procedure raises an alarm at time $\hat \tau > \tau$. First we estimate the IC mean $\bm \mu_0$ by $\hat {\bm \mu}_0$. Then we compute the residual matrices on a window of OOC images $\{\mathbf R_{\hat\tau+i} = \mathbf X_{\hat\tau+i} - \hat{\bm \mu}_0, \ i = 1,\ldots, w\}$. We perform Hoyer index on the sample average of the residual matrices.
Algorithm \ref{algorithm:quick diagnosis} displays the concrete steps. 


\begin{algorithm}[H]
\caption{Scheme to compute Hoyer index online}
\label{algorithm:quick diagnosis}
{\it Input}: Window lengths $w_0,w$, IC images $\{\mathbf X_{-i+1}:\ i=1,\ldots,w_0\}$, OOC images $\{\mathbf X_{\hat\tau+i:}\ i=1,\ldots,w\}$.
\begin{algorithmic}[1]
    \State   Estimate IC mean $\hat {\bm\mu}_0 = \sum_{i=1}^{w_0} \mathbf X_{-i+1}/w_0$. 
    \State  Compute residual matrices for OOC images: 
    $$
    \mathbf R_{\hat\tau+i} = \mathbf X_{\hat\tau+i}-\hat{\bm\mu}_0, \quad i=1,\ldots,w.
    $$ 
    \State 	Perform
    $h\left(\sum_{i=1}^{w} \mathbf R_{\hat\tau+i}/w\right)$ with the Hoyer index in \eqref{formula:hoyer index}.
\end{algorithmic}
\end{algorithm}


\section{Theoretical Analysis}\label{sec:theory}
In this section, we discuss the condition that Hoyer index is unbiased and the remedy when the Hoyer index is deviated by the noise. 

Algorithm \ref{algorithm:quick diagnosis} could work when there is no noise or the noise to signal ratio is small. When the noise is large, we first consider a simple scenario where $\{\mathbf X_t:\ t>\tau\}$ are i.i.d.. 
\begin{lemma}
\label{lem:ooc-iid}
Suppose $\{\mathbf X_t:t>\tau\}$ are i.i.d., then we have
\begin{equation}
    h\left(\sum_{i=1}^{w} \mathbf R_{\hat\tau+i}/w\right)\xrightarrow{a.s.}{} h(\mathbf A) \quad \text{as }w\to\infty.
\end{equation}
\end{lemma}
Lemma \ref{lem:ooc-iid} is simply proved by the Law of Large Number (LLN). However, the applicable conditions narrow with the i.i.d. assumption on the OOC images. Afterwards, we consider a single OOC image with large dimensions to circumvent any assumption on the distributions of OOC images.

Lemma \ref{lem:sparsity of white noise} shows that the Hoyer index deems a white noise matrix as very sparse in an asymptotic sense. In other words, the white noise matrix is non-informative with its Hoyer index nearly one. 
\begin{lemma}[Convergence rate of white noise matrix]
\label{lem:sparsity of white noise}
Suppose the entries of a white noise matrix $\mathbf e\in\mathbb R^{{p_1}\times {p_2}}$ are i.i.d. with zero mean and variances $\sigma^2$. Then we have
\begin{equation}
    1-h(\mathbf e) = \operatorname{O}_{a.s.}\left(\sqrt{\frac{\log\log p_1p_2}{p_1p_2}}\right).
\end{equation}
\end{lemma}

Theorem~\ref{thm:affection of white noise} shows that the white noise laid on the anomaly could deviate the Hoyer index.
\begin{theorem}\label{thm:affection of white noise}
We denote the asymptotic average magnitude of the anomaly and the square of the anomaly as $\bar a = \lim_{{p_1},{p_2}\to\infty} \left|\sum_{i=1}^{p_1}\sum_{j=1}^{p_2} A_{ij}\right|/p_1p_2$ and $\bar {a^2} = \lim_{{p_1},{p_2}\to\infty} \|\mathbf A\|_F^2/p_1p_2>0$, respectively.
For a white noise matrix $\mathbf e$ in Lemma \ref{lem:sparsity of white noise}, 
we have
\begin{equation}
\begin{aligned}
    h(\mathbf A+\mathbf e) - h(\mathbf A) \xrightarrow{a.s.}{}& \frac{\bar{a} \sigma^2}{\sqrt{\bar{a^2}\left(\bar{a^2}+\sigma^2\right)}\left(\sqrt{\bar{a^2}}+\sqrt{\bar{a^2}+\sigma^2}\right)}\\
    &\quad\text{as}\quad {p_1},{p_2} \to \infty.
\end{aligned}
\end{equation}
\end{theorem}

Theorem \ref{thm:affection of white noise} analytically characterize the affection of the noise matrix $\mathbf e_t$ on evaluation of the sparsity of the anomaly. The bias is related to the signals aggregated from $\mathbf A$ and the noise level. As a sanity check, when the anomaly vanishes, the term $\bar{a} = 0$ and then the asymptotic difference between $h(\mathbf A + \mathbf e)$ and $h(\mathbf A)$ also vanishes. This result coincides with Lemma \ref{lem:sparsity of white noise}, namely a white noise matrix is asymptotically as sparse as a blank matrix. If the noise vanishes and $\sigma^2$ decreases to 0, the difference between $h(\mathbf R_t)$ and $h(\mathbf A)$ also converges to 0, which is a naive circumstance. If the noise diverges to infinite, we can have the below corollary.
\begin{corollary}
\label{coro:divergent noise}
With conditions in Theorem \ref{thm:affection of white noise}, if as ${p_1},{p_2}\to\infty$, the individual variance $\sigma^2\to\infty$, then we have 
\begin{equation}
    h(\mathbf A+\mathbf e)\xrightarrow{a.s.}{} 1 \quad\text{as}\quad {p_1},{p_2} \to \infty.
\end{equation}
\end{corollary}
Corollary \ref{coro:divergent noise} illustrates that with sufficiently large noise, the anomaly can be dominated in the noise. In such a case, for any dense anomaly, the conventional Hoyer index can indicate it to be sparse. This is consistent with the visualizations in Figure \ref{fig:white noise}, where the growing noise gradually overwhelms the ring-like anomaly pattern. \begin{figure}[ht!]
\centering 
\begin{subfigure}[h]{0.3\linewidth}
\includegraphics[width=\linewidth]{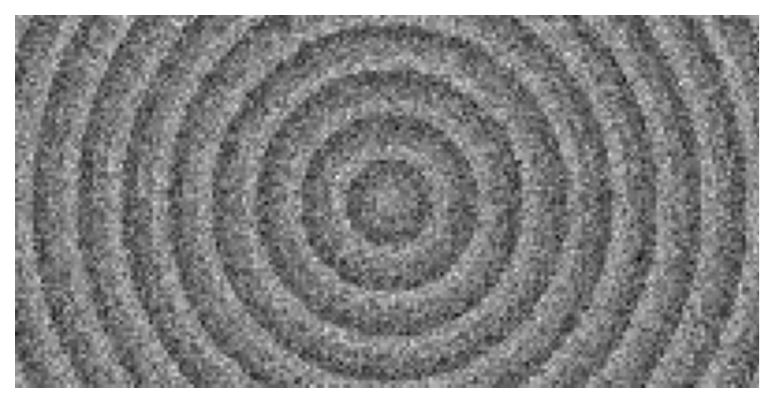}
\caption{Low noise}
\end{subfigure}
\begin{subfigure}[h]{0.3\linewidth}
\includegraphics[width=\linewidth]{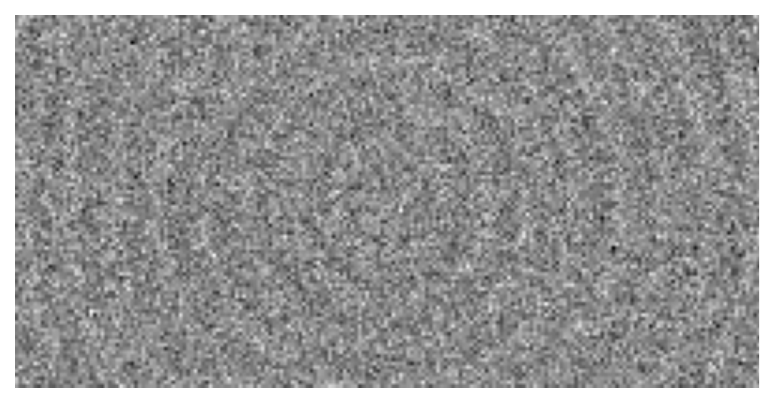}
\caption{Medium noise}
\end{subfigure}
\begin{subfigure}[h]{0.3\linewidth}
\includegraphics[width=\linewidth]{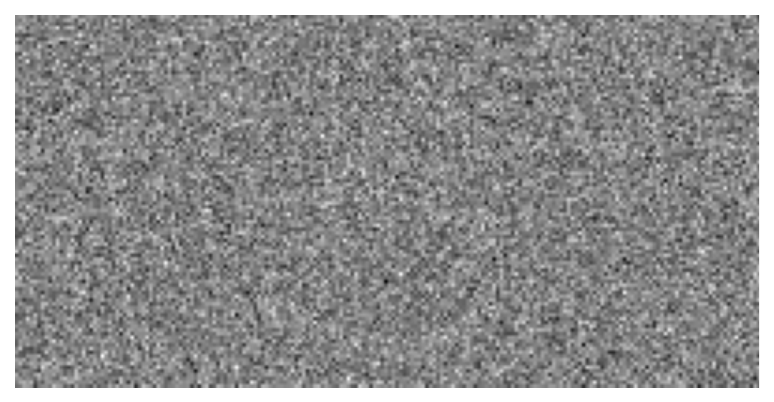}
\caption{High noise}
\end{subfigure}
\caption{Visualizations showing growing noise overwhelming the ring-like anomaly pattern. }
\label{fig:white noise}
\end{figure}

Since the conventional Hoyer index is not consistent, we define the corrected Hoyer index on OOC residual matrices. Namely, for $t>\tau$, we define
\begin{equation}\label{formula:corrected Hoyer index}
    g(\mathbf R_t) = h(\mathbf R_t) - \frac{\bar{a} \sigma^2}{\sqrt{\bar{a^2}\left(\bar{a^2}+\sigma^2\right)}\left(\sqrt{\bar{a^2}}+\sqrt{\bar{a^2}+\sigma^2}\right)}.
\end{equation}
With LLN on $\hat{\bm\mu}_0$ and Theorem \ref{thm:affection of white noise}, we have the consistency of $g(\mathbf R_t)$.
\begin{corollary}\label{coro:consistent remedy}
    With conditions in Theorem \ref{thm:affection of white noise}, the corrected Hoyer index $g(\mathbf R_t)$ is a consistent estimator of $h(\mathbf A)$ w.r.t. $w_0,p_1,p_2$. 
\end{corollary}
With Corollary \ref{coro:consistent remedy}, we summarize the corrected scheme to compute the sparsity in Algorithm \ref{algorithm:corrected quick diagnosis}. 

\begin{algorithm}[H]
\caption{A corrected scheme to compute Hoyer index}
\label{algorithm:corrected quick diagnosis}
{\it Input}: Window lengths $w_0$, IC images $\{\mathbf X_{-i+1}:\ i=1,\ldots,w_0\}$, OOC image $\mathbf X_{\hat\tau+1}$.
\begin{algorithmic}[1]
    \State   Estimate IC mean $\hat {\bm\mu}_0 = \sum_{i=1}^{w_0} \mathbf X_{-i+1}/w_0$. 
    \State   Compute IC residual matrices: 
    $$
    \mathbf R_{-i+1} = \mathbf X_{-i+1}-\hat{\bm\mu}_0,\quad i = 1,\ldots, w_0.
    $$ 
    \State Estimate $\sigma^2$ by the sample variance of all entries in IC images $\{\mathbf R_{-i+1}:\ i=1,\ldots,w_0\}$, denoted as $\hat\sigma^2$.
    \State   Compute the OOC residual matrix: 
    $$
    \mathbf R_{\hat \tau} = \mathbf X_{\hat\tau}-\hat{\bm\mu}_0.
    $$ 

    \State Estimate $\bar{a}$ and $\bar{a^2}$ with:  
    $$
    \hat a = \left|\sum_{i=1}^{p_1}\sum_{j=1}^{p_2} R_{\hat\tau,i,j}\right|/p_1p_2,\quad \widehat {a^2} = \|\mathbf R_{\hat\tau}\|_F^2/p_1p_2.
    $$
    \State Compute
    $g(\mathbf R_{\hat\tau})$ with the corrected Hoyer index in \eqref{formula:corrected Hoyer index}.
\end{algorithmic}
\end{algorithm}

\section{Simulated Experiments}\label{sec:simulation}
In this section, we design simulated experiments for two purposes. First, the corrected Hoyer index is able to estimate the sparsity of the anomaly despite of the noise. Second, it is consistent w.r.t. the dimensions. 

\subsection{Robustness}
We set the dimensions of the anomaly to be ${p_1} = 100, {p_2} = 200$. For the dense anomaly, we design it to be: 
\begin{equation*}
    A_{ij} = \lfloor (j-1)/50 \rfloor, \quad i = 1,\ldots,100, \ j = 1,\ldots,200. 
\end{equation*}
For the sparse anomaly, we define as:
\begin{equation*}
    A_{ij} = 5\mathbb I (50\leq j <60), \quad i = 1,\ldots,100, \ j = 1,\ldots,200. 
\end{equation*}
Figure \ref{fig:dense} and \ref{fig:sparse} visualize the two types of anomalies. Since the estimation of $\hat{\bm\mu}_0$ is not the focus of the paper, given the anomaly $\mathbf{A}$ and the standard deviation $\sigma\in\{0.5,1.0,\ldots, 5.5,6.0\}$, we directly simulate IC and OOC residual matrices: 
\begin{equation}\label{formula:pair of residual}
    \mathbf{R}_t = \mathbb{I}(t>0)\mathbf{A}+ \mathbf{e}_t,\quad t=-199,\ldots,200.
\end{equation}
where $\mathbf e_t$ are white noise matrices with entries i.i.d. following normal $N(0,\sigma^2)$. 
For each OOC residual matrix in $\{\mathbf R_t, \ t>0\}$, we compute its corrected Hoyer index through Algorithm \ref{algorithm:corrected quick diagnosis} with the window size $w_0=200$. We quantify the errors of corrected Hoyer index by the interval:
\begin{equation}\label{formula:output1}
    [m_\varepsilon - 1.96 \sigma_\varepsilon, m_\varepsilon +  1.96 \sigma_\varepsilon], 
\end{equation}
where $\varepsilon_t = \left|g\left(\mathbf R_t\right)-h(\mathbf A)\right|, m_\varepsilon = {\sum_{t=1}^{200} \varepsilon_t}/{200}$ and $\sigma^2_\varepsilon = {\sum_{t=1}^{200}\left(\varepsilon_t-m_\varepsilon\right)^2}/{(200-1)}$.

The result is shown in Figure \ref{fig:robustness}. The fact that Hoyer index takes value within $[0,1]$ naturally justifies the absolute error $\varepsilon_t$. Hence there is no need to further report the relative error of corrected Hoyer index. Indeed, relative error will bias our judgement on the recovery of the underlying sparsity. For an example, if the underlying anomaly is dense and has a small Hoyer index, then the relative error turns out big. Both of the sample mean of the error of corrected Hoyer indices are under $0.08$ uniformly, validating the robustness of the corrected Hoyer index against the noise level. As the noise grows, we visualize the residual matrices as it appears to the naked eye in Figure \ref{fig:anomaly vs variances}. When $\sigma = 6$, it becomes hard for human to instantly judge whether the anomaly is sparse or dense. However, such issue is resolved using corrected Hoyer index. 

\begin{figure}[ht!]
\centering 
\begin{subfigure}[h]{0.45\linewidth}
\includegraphics[width=\linewidth]{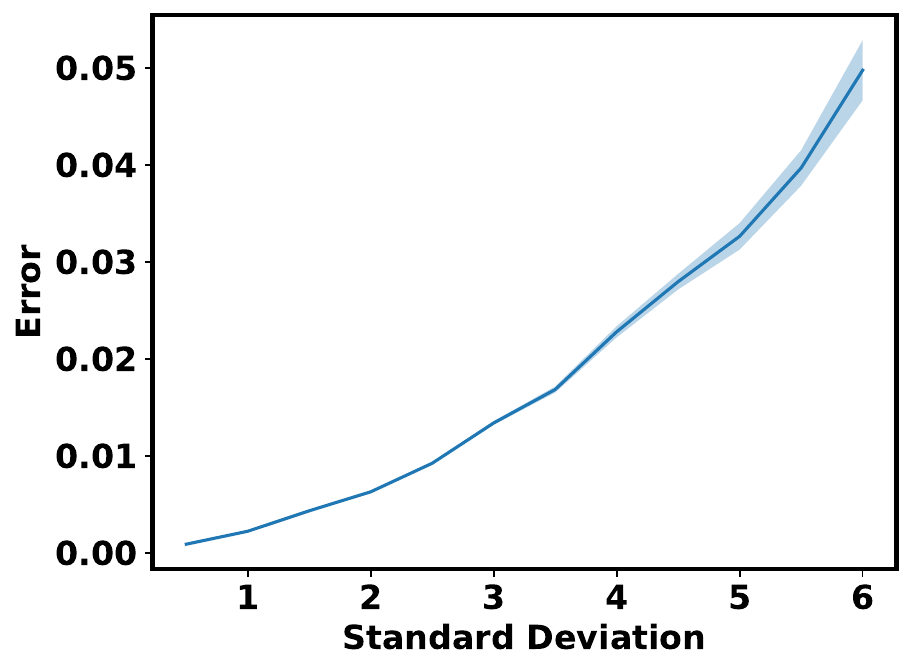}
\caption{dense anomaly}
\end{subfigure}
\begin{subfigure}[h]{0.45\linewidth}
\includegraphics[width=\linewidth]{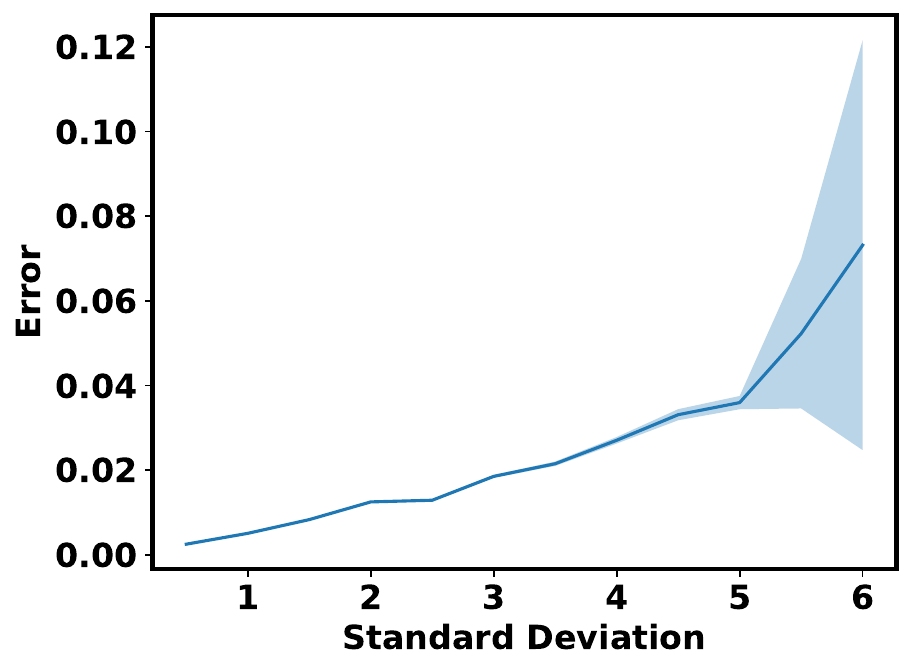}
\caption{sparse anomaly}
\end{subfigure}
\caption{Error $m_\varepsilon$ against Standard Deviation $\sigma$ curves for dense and sparse anomalies with fixed dimension ${p_1}=100,{p_2}=200$. (a) and (b) correspond to the dense and sparse anomaly, respectively. The blue line reports the error $m_\varepsilon$ w.r.t. $\sigma$. The light blue bands visualize the interval $[m_\varepsilon - 1.96 \sigma_\varepsilon, m_\varepsilon +  1.96 \sigma_\varepsilon]$ in \eqref{formula:output1}. }
\label{fig:robustness}
\end{figure}

\begin{figure}[ht!]
\centering 
\begin{subfigure}[h]{0.24\linewidth}
\includegraphics[width=\linewidth]{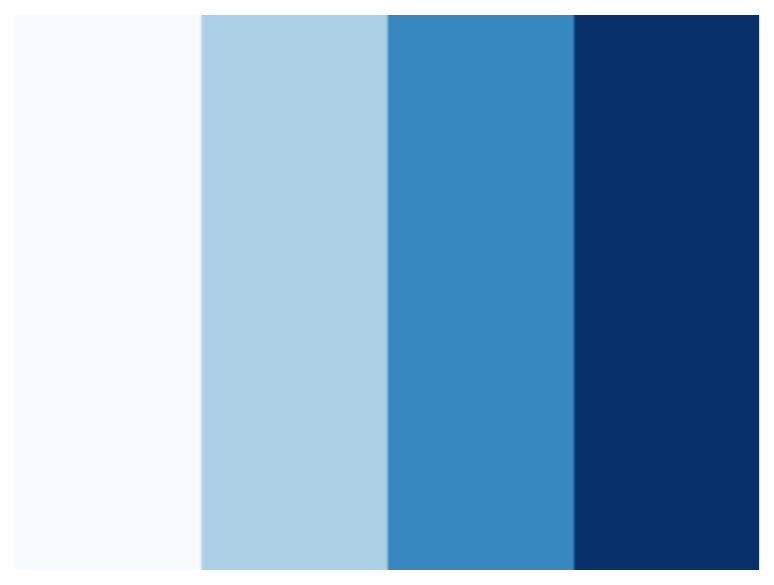}
\caption{dense shift}
\label{fig:dense}
\end{subfigure}
\begin{subfigure}[h]{0.24\linewidth}
\includegraphics[width=\linewidth]{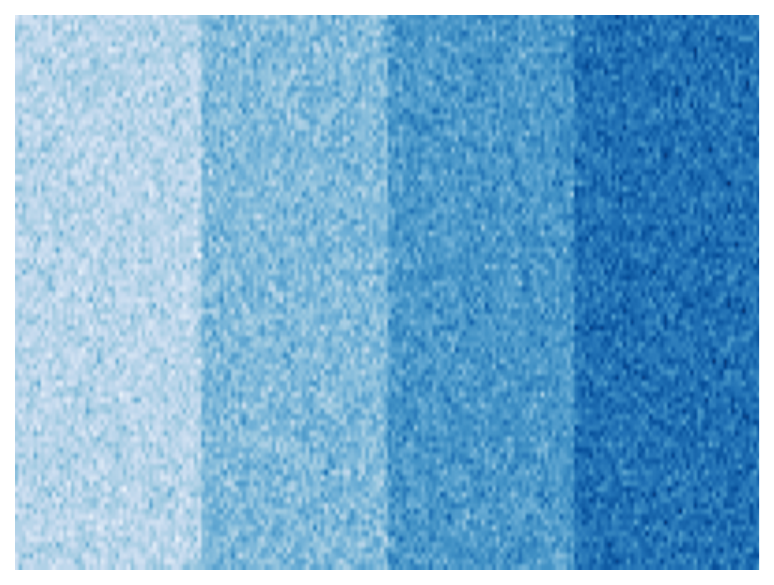}
\caption{$\sigma = 0.5$}
\end{subfigure}
\begin{subfigure}[h]{0.24\linewidth}
\includegraphics[width=\linewidth]{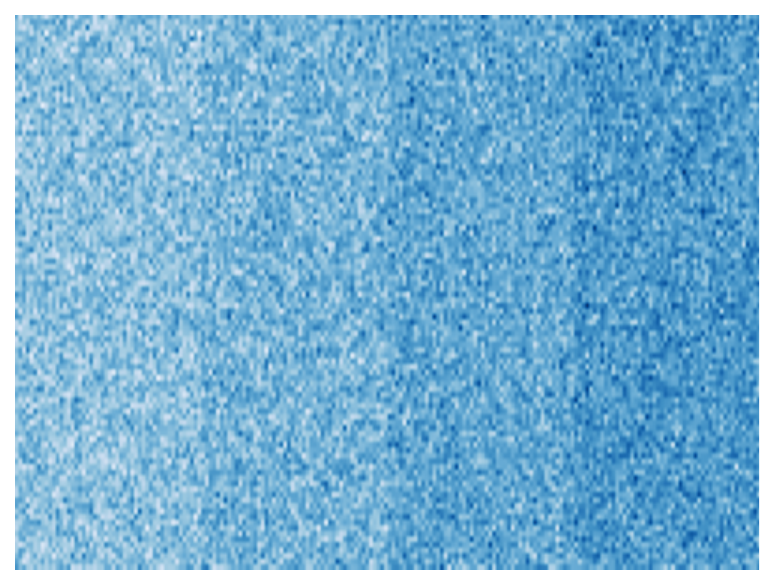}
\caption{$\sigma = 2$}
\end{subfigure}
\begin{subfigure}[h]{0.24\linewidth}
\includegraphics[width=\linewidth]{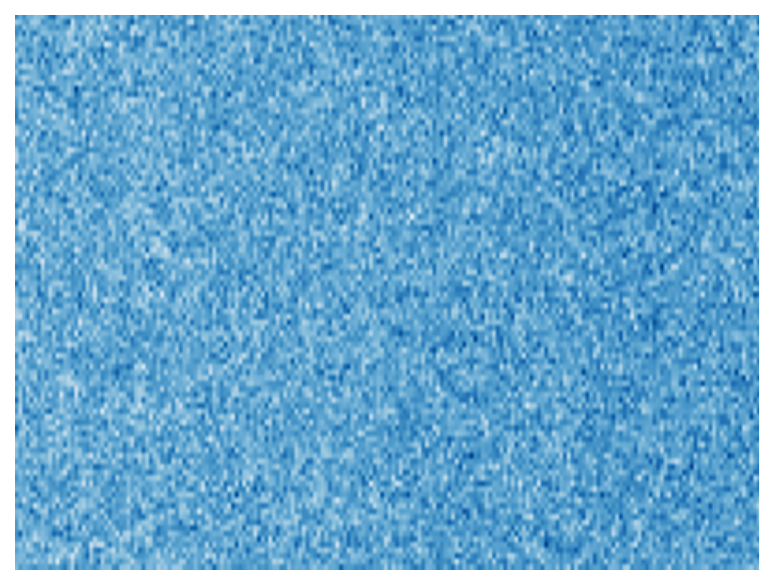}
\caption{$\sigma = 6$}
\end{subfigure}

\begin{subfigure}[h]{0.24\linewidth}
\includegraphics[width=\linewidth]{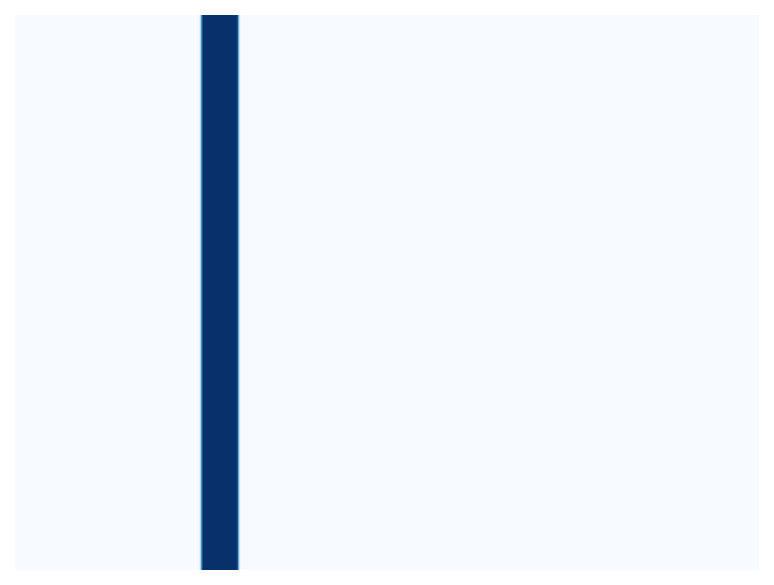}
\caption{sparse shift}
\label{fig:sparse}
\end{subfigure}
\begin{subfigure}[h]{0.24\linewidth}
\includegraphics[width=\linewidth]{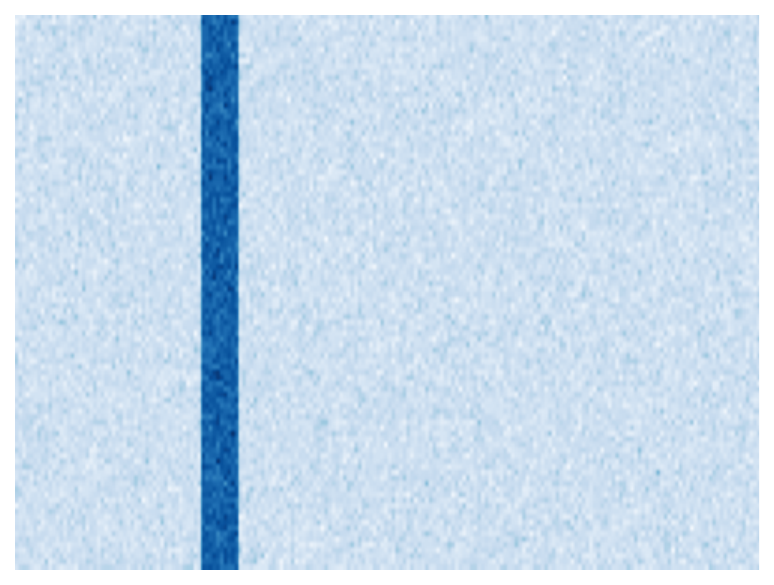}
\caption{$\sigma = 0.5$}
\end{subfigure}
\begin{subfigure}[h]{0.24\linewidth}
\includegraphics[width=\linewidth]{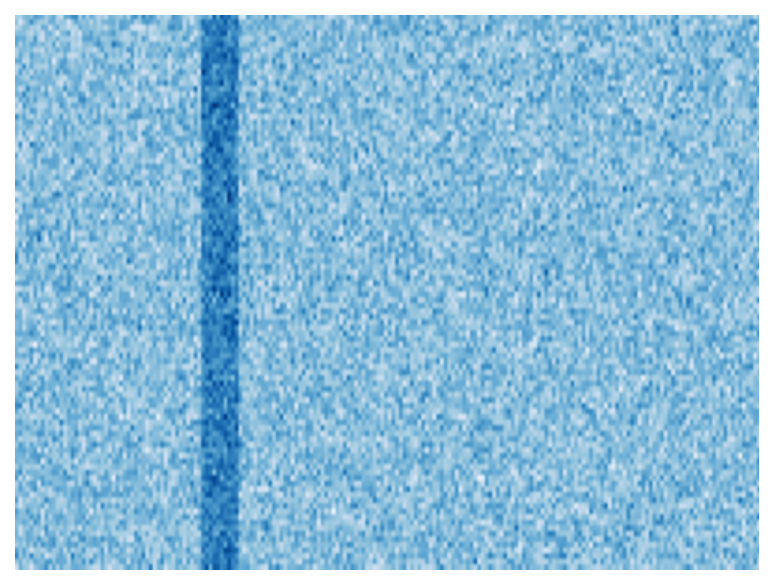}
\caption{$\sigma = 2$}
\end{subfigure}
\begin{subfigure}[h]{0.24\linewidth}
\includegraphics[width=\linewidth]{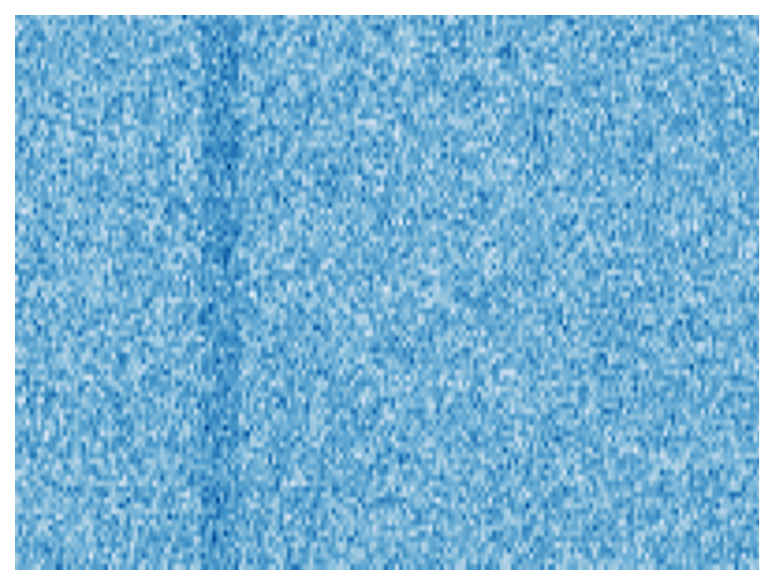}
\caption{$\sigma = 6$}
\end{subfigure}
\caption{Dense or sparse anomaly interrupted by different degrees of white noise matrices $\mathbf e\sim \mathcal{MN}(0,\sigma I_w, \sigma I_p)$. Dense anomaly with $\sigma \in\{0.5,2,5\}$ are shown in (a), (b) and (c), respectively. Sparse anomaly with $\sigma \in\{0.5,2,5\}$ are shown in (d), (e) and (f), respectively. }
\label{fig:anomaly vs variances}
\end{figure}
\subsection{Consistency}
We empirically validate the consistency in Corollary \ref{coro:consistent remedy}. Now we fix the standard deviation $\sigma = 3$ and slides the dimension of the matrix. Consider the basic dimension $\underline{p}_1 = 1, \underline{p}_2 = 2$. We magnify the basic dimension with a multiplier $c\in\{10,20,\ldots, 100\}$. To be precise, given $c$, the dimensions are ${p_1} = c\underline{p}_1, {p_2} = c\underline{p}_2$. Correspondingly, the design on the anomalies have the form:
\begin{equation}\label{formula:anomaly for consistency}
    \begin{aligned}
        &\text{Dense: } A_{ij} = \left\lfloor \frac{4(j-1)}{cp_0} \right\rfloor, \quad i = 1,\ldots,cw_0, \ j = 1,\ldots,cp_0, \\
    &\text{Sparse: } A_{ij} = 5\mathbb I \left(\frac{cp_0}{4}\leq j <\frac{cp_0}{4}+\frac{c}{10}\right),\\
    &\quad \quad \quad \quad i = 1,\ldots,cw_0, \ j = 1,\ldots,cp_0. 
    \end{aligned}
\end{equation}
Both types of the anomaly inherit the patterns in Figure \ref{fig:dense} and \ref{fig:sparse}, with different dimensions. Given the anomaly matrix and a magnification $c$, we simulate the sequence containing IC and OOC matrices with the model \eqref{formula:pair of residual}.
$$
    \mathbf{R}_t = \mathbb{I}(t>0)\mathbf{A}+ \mathbf{e}_t,\quad t=-199,\ldots,200.
$$ 
except for different dimensions.
Still, we report the errors of corrected Hoyer index by the equation \eqref{formula:output1} in terms of banded curves in Figure \ref{fig:consistency}. Note that the standard deviation $\sigma$ has been fixed and the statistics is now dependent on the dimension magnification $c$ instead. In both types, as the dimensions grow, both of errors and fluctuations on errors rapidly decay. The empirical observations coincides with Corollary \ref{coro:consistent remedy}. 

\begin{figure}[ht!]
\centering 
\begin{subfigure}[h]{0.45\linewidth}
\includegraphics[width=\linewidth]{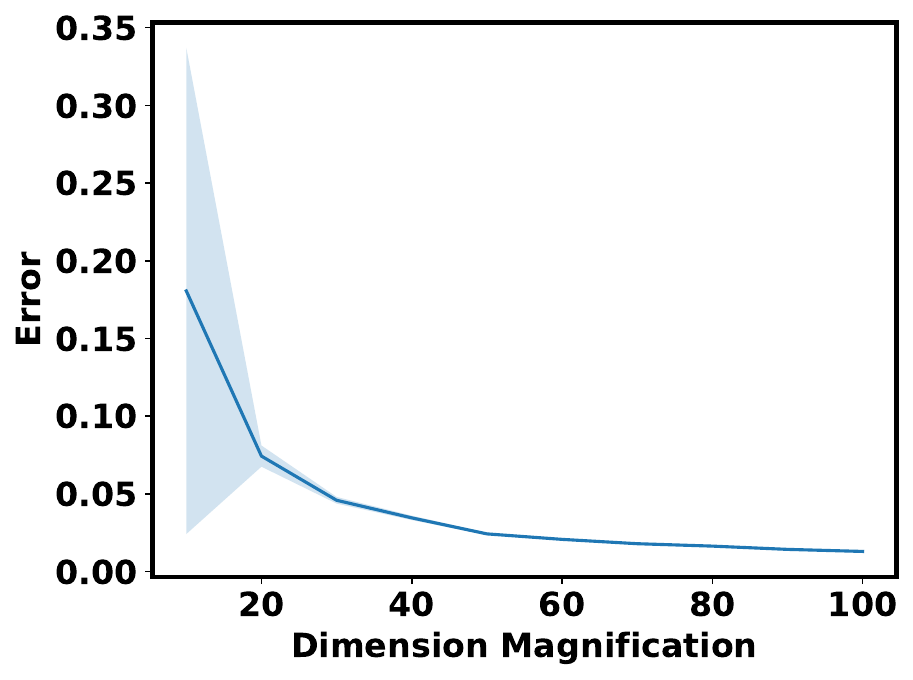}
\caption{dense anomaly}
\end{subfigure}
\begin{subfigure}[h]{0.45\linewidth}
\includegraphics[width=\linewidth]{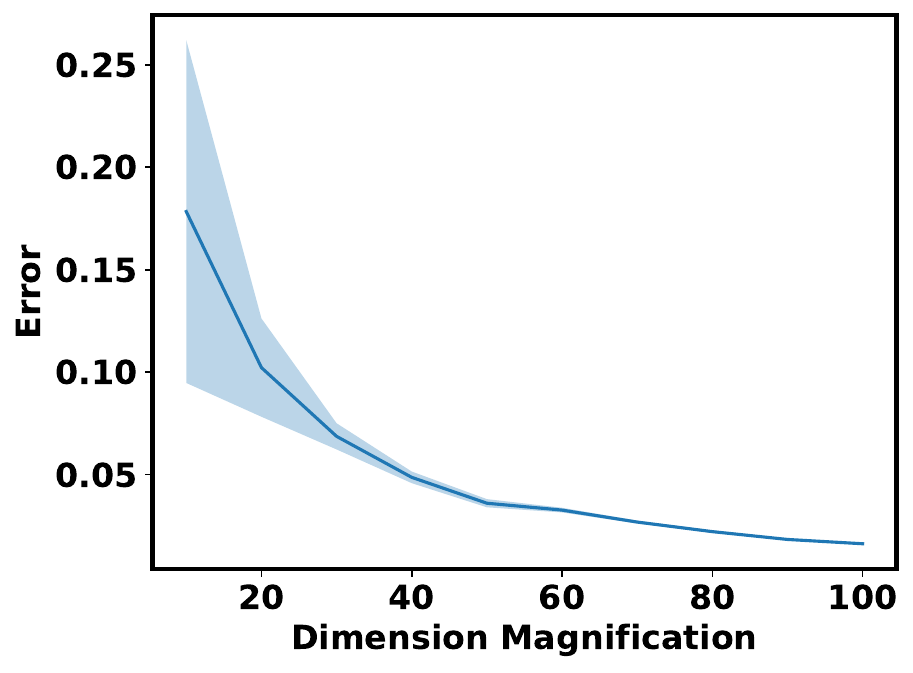}
\caption{sparse anomaly}
\end{subfigure}
\caption{Error $m_\varepsilon$ against Dimension Magnification $c$ curves for dense and sparse anomalies with fixed standard deviation $\sigma = 3$. (a) and (b) correspond to the dense and sparse anomaly, respectively. The magnitudes of blue lines and light blue bands have similar meaning as in Figure \ref{fig:robustness}.  }
\label{fig:consistency}
\end{figure}
\section{Real Data Experiments}\label{sec:realdata}
In this section, we apply the proposed scheme to two real data sets. We will see that the proposed module accurately reflect the underlying sparsity of the anomaly and meanwhile provide a baseline to aid change-point detection. 
\subsection{Abnormal Crowd Activity}
In the first example, we perform the corrected Hoyer index onto a video during which an abnormal crowd activity happened. The raw data is available at \url{http://mha.cs.umn.edu/proj_events.shtml#crowd}. We select one video from the multiple of them. In the first half of video, people are talking normally in the hall. Near the end of the video, something abnormal happens and the crowd flees the hall. After simple processing, the selected video consists of 578 grayscale frames. Each frame has dimensions ${p_1} = 130$ and ${p_2} = 320$. Two of the frames, one normal and another abnormal, have been displayed in the first row of Figure \ref{fig:illustrative cases}. 

For abnormal crowd activity data, we select IC window size $w_0 = 100$ and vary the potential change-point $\hat\tau\in\{201,\ldots,578\}$. Algorithm \ref{algorithm:corrected quick diagnosis} generates a series of corrected Hoyer indices from $\mathbf{R}_{\hat\tau}$. Figure \ref{fig:crowd} computes the behavior of the corrected Hoyer index against time. At first, the corrected Hoyer index is around 1, indicating the residual signal is very sparse, namely no anomaly or sparse anomaly exists. In period $t\in[250,450]$, the index reports a valley-shaped trend that first decreases before $t=350$ and afterwards then increases back to 1. We pause at time $t = 275, 350$ and $450$ and pull out the images. A local anomaly that in the lower left corner, the door opens with light shedding in and then door closes with light vanishing,  happens during period $[250,450]$. When the light sheds in, the white spot in the picture decreases the sparsity of image and thus lowers the corrected Hoyer index. When the light disappears, the process is the opposite. During the local ``light-in-out" process, the crowd remains normal and the decrease of the corrected Hoyer index is slight. At $t=500$, the anomaly occurs and people flee from the scene. Simultaneously, the corrected Hoyer index sharply decreases to $0.4$ until the time $t=550$, that the hall is empty. The empirical study shows the abnormal crowd activity is indeed dense in terms of low Hoyer index. 
\begin{figure}[ht!]
\centering 
\begin{subfigure}[h]{0.9\linewidth}
\includegraphics[width=\linewidth]{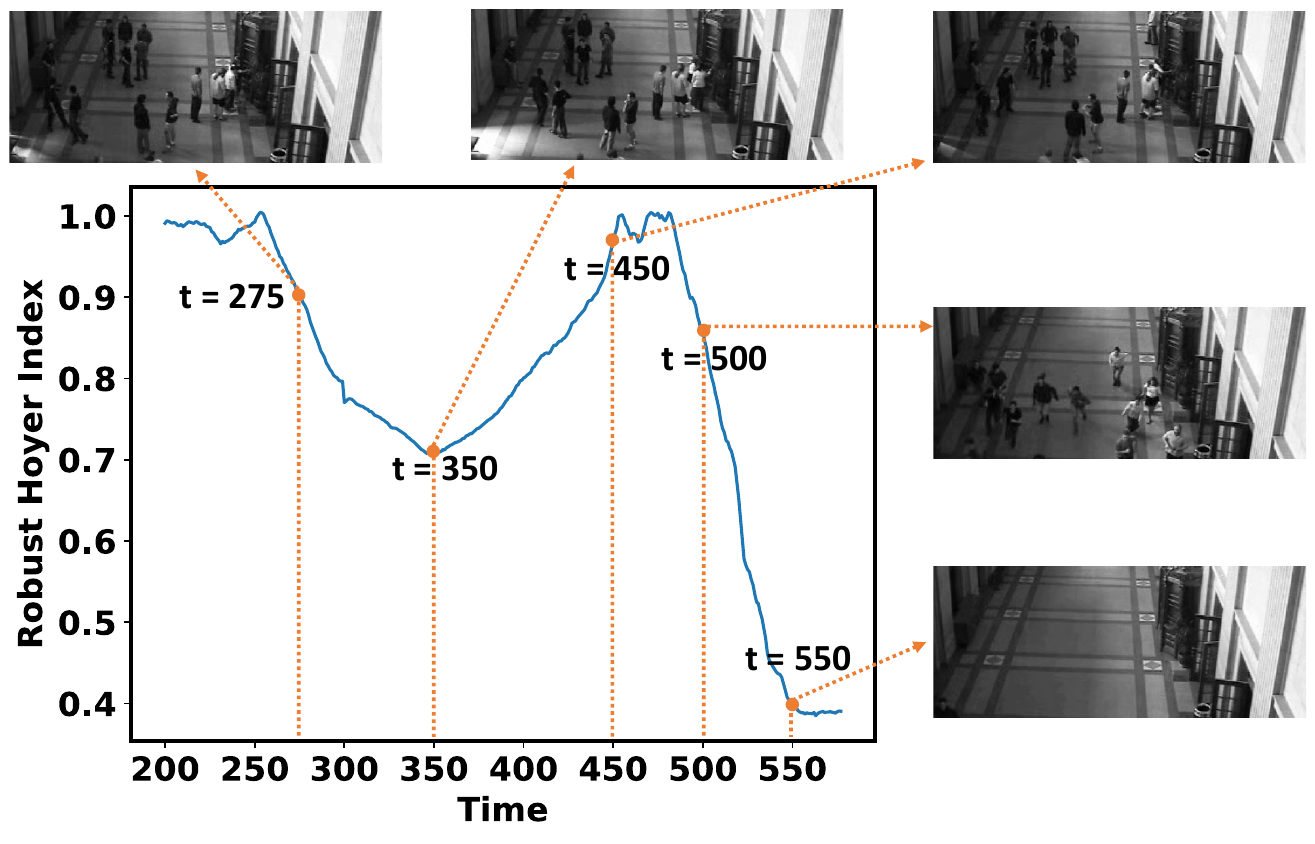}
\end{subfigure}
\caption{Behavior of the corrected Hoyer index againt time for Abnormal Crowd Activity data. }
\label{fig:crowd}
\end{figure}
\subsection{Solar Flare}
In the second example, we study with a stream of solar images. At some period, violent outburst of solar flare can appear, which may potentially damage the power-grids and cause financial loss. This data is publicly available at \url{http://nislab.ee.duke.edu/MOUSSE/index.html}. As a convention, we select $300$ consecutive frames, size of each are $232\times 292$. The normal and abnormal observations have been shown in the second row of Figure \ref{fig:illustrative cases}. 

For solar flare data, we pick IC window size $w_0 = 100$ and then slide $\hat\tau\in\{201,202,\ldots,300\}$. During the process, Figure \ref{fig:solar} shows that the corrected Hoyer index maintains a relatively high magnitude, indicating that the underlying anomaly is sparse. We take three snapshots at time $t = 195,230,280$. At $t=195$, the solar activity is normal. Then at $t=230$, a solar flare happens but the corrected Hoyer index is still high. When $t=280$, the solar activity calms to normal. 
\begin{figure}[ht!]
\centering 
\begin{subfigure}[h]{0.9\linewidth}
\includegraphics[width=\linewidth]{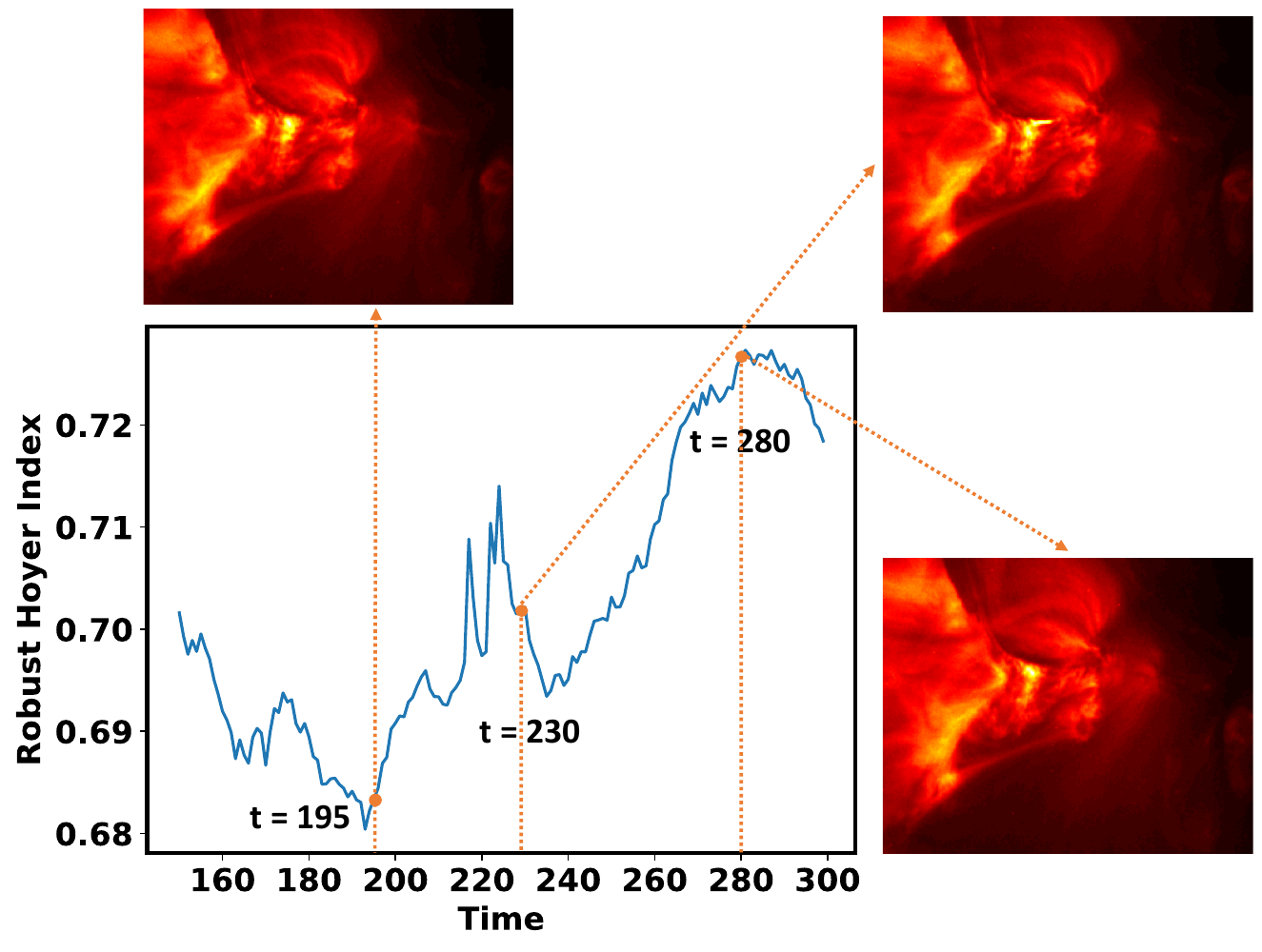}
\end{subfigure}
\caption{Behavior of the corrected Hoyer index against time for Solar Flare data. }
\label{fig:solar}
\end{figure}

Another interesting observation is displayed in Figure \ref{fig:similar pattern}. Though we construct the corrected Hoyer index with the main goal to measure the sparsity of the underlying anomaly, we also observe a by-product that the corrected Hoyer index can reveal similar pattern with the testing statistics proposed in the advanced work \cite{yan2018real}. Thus the corrected Hoyer index has the potential to become a informative feature and aid the monitoring procedure. 
\begin{figure}[ht!]
\centering 
\begin{subfigure}[h]{0.9\linewidth}
\includegraphics[width=\linewidth]{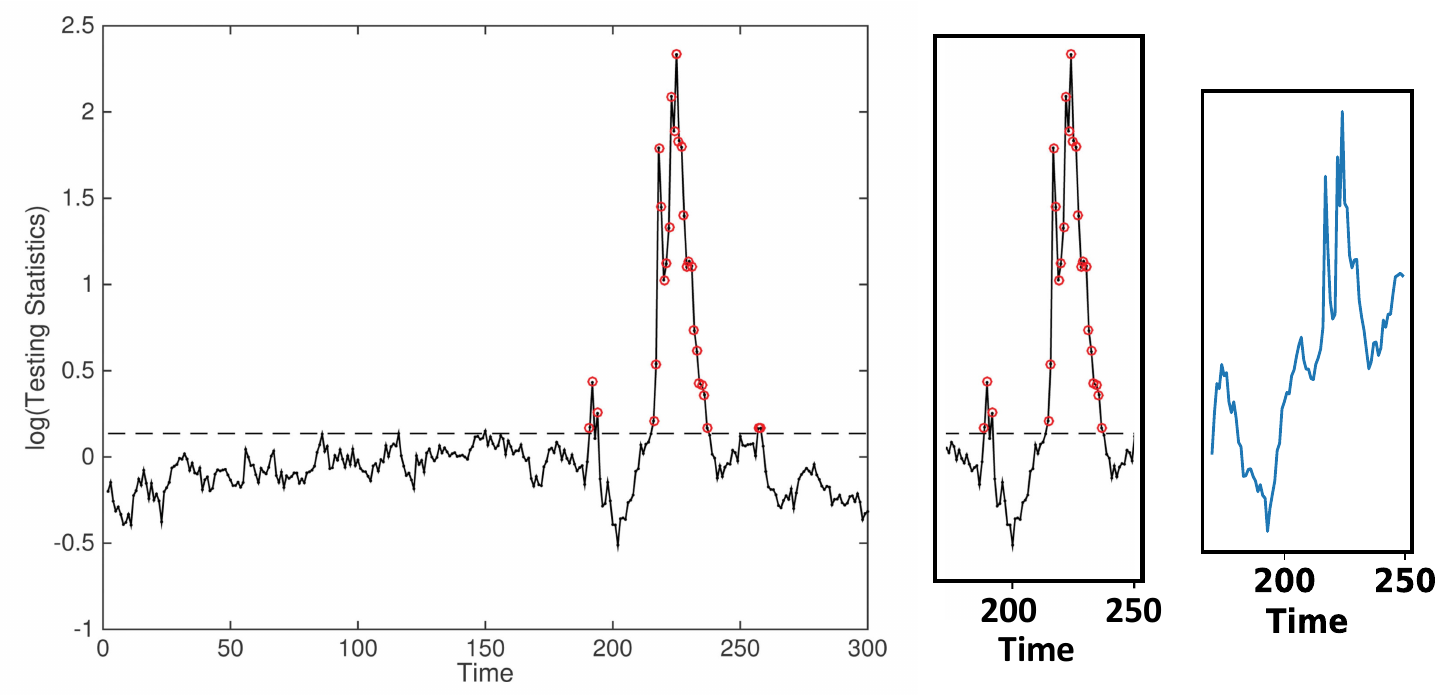}
\end{subfigure}
\caption{The left panel is quoted from Figure 7 in \cite{yan2018real}.The middle panel is the screenshot from the left panel with time ranging in $[170,250)$. The right panel is the screenshot from Figure \ref{fig:solar} with time ranging in $[170,250)$. }
\label{fig:similar pattern}
\end{figure}

\section{Conclusions}\label{sec:conclusion}
In the paper, we propose the corrected Hoyer index which can quickly estimate the sparsity of any OOC image against noise. We theoretically prove the consistency of the proposed index. We demonstrate its robustness against noise level and consistency w.r.t. the dimensions with the simulated experiments. Also the corrected Hoyer index coincides with the prior knowledge of the sparsity of the anomaly in the real image sequences.

\newpage
\bibliographystyle{IEEEtran}
\bibliography{main}

\begin{thebibliography}{10}
\providecommand{\url}[1]{#1}
\csname url@samestyle\endcsname
\providecommand{\newblock}{\relax}
\providecommand{\bibinfo}[2]{#2}
\providecommand{\BIBentrySTDinterwordspacing}{\spaceskip=0pt\relax}
\providecommand{\BIBentryALTinterwordstretchfactor}{4}
\providecommand{\BIBentryALTinterwordspacing}{\spaceskip=\fontdimen2\font plus
\BIBentryALTinterwordstretchfactor\fontdimen3\font minus
  \fontdimen4\font\relax}
\providecommand{\BIBforeignlanguage}[2]{{%
\expandafter\ifx\csname l@#1\endcsname\relax
\typeout{** WARNING: IEEEtran.bst: No hyphenation pattern has been}%
\typeout{** loaded for the language `#1'. Using the pattern for}%
\typeout{** the default language instead.}%
\else
\language=\csname l@#1\endcsname
\fi
#2}}
\providecommand{\BIBdecl}{\relax}
\BIBdecl

\bibitem{qiu2018some}
P.~Qiu, ``Some perspectives on nonparametric statistical process control,''
  \emph{Journal of Quality Technology}, vol.~50, no.~1, pp. 49--65, 2018.

\bibitem{reis2017industrial}
M.~S. Reis and G.~Gins, ``Industrial process monitoring in the big
  data/industry 4.0 era: From detection, to diagnosis, to prognosis,''
  \emph{Processes}, vol.~5, no.~3, p.~35, 2017.

\bibitem{wright2008robust}
J.~Wright, A.~Y. Yang, A.~Ganesh, S.~S. Sastry, and Y.~Ma, ``Robust face
  recognition via sparse representation,'' \emph{IEEE transactions on pattern
  analysis and machine intelligence}, vol.~31, no.~2, pp. 210--227, 2008.

\bibitem{wang2021deep}
M.~Wang and W.~Deng, ``Deep face recognition: A survey,''
  \emph{Neurocomputing}, vol. 429, pp. 215--244, 2021.

\bibitem{mairal2007sparse}
J.~Mairal, M.~Elad, and G.~Sapiro, ``Sparse representation for color image
  restoration,'' \emph{IEEE Transactions on image processing}, vol.~17, no.~1,
  pp. 53--69, 2007.

\bibitem{aharon2008sparse}
M.~Aharon and M.~Elad, ``Sparse and redundant modeling of image content using
  an image-signature-dictionary,'' \emph{SIAM Journal on Imaging Sciences},
  vol.~1, no.~3, pp. 228--247, 2008.

\bibitem{wright2010sparse}
J.~Wright, Y.~Ma, J.~Mairal, G.~Sapiro, T.~S. Huang, and S.~Yan, ``Sparse
  representation for computer vision and pattern recognition,''
  \emph{Proceedings of the IEEE}, vol.~98, no.~6, pp. 1031--1044, 2010.

\bibitem{herath2017going}
S.~Herath, M.~Harandi, and F.~Porikli, ``Going deeper into action recognition:
  A survey,'' \emph{Image and vision computing}, vol.~60, pp. 4--21, 2017.

\bibitem{leung2008sparse}
K.~E. Leung, M.~van Stralen, A.~Nemes, M.~M. Voormolen, G.~van Burken, M.~L.
  Geleijnse, F.~J. Ten~Cate, J.~H. Reiber, N.~de~Jong, A.~F. van~der Steen
  \emph{et~al.}, ``Sparse registration for three-dimensional stress
  echocardiography,'' \emph{IEEE transactions on medical imaging}, vol.~27,
  no.~11, pp. 1568--1579, 2008.

\bibitem{ferrante2017slice}
E.~Ferrante and N.~Paragios, ``Slice-to-volume medical image registration: A
  survey,'' \emph{Medical image analysis}, vol.~39, pp. 101--123, 2017.

\bibitem{yan2018real}
H.~Yan, K.~Paynabar, and J.~Shi, ``Real-time monitoring of high-dimensional
  functional data streams via spatio-temporal smooth sparse decomposition,''
  \emph{Technometrics}, vol.~60, no.~2, pp. 181--197, 2018.

\bibitem{yan2022real}
H.~Yan, M.~Grasso, K.~Paynabar, and B.~M. Colosimo, ``Real-time detection of
  clustered events in video-imaging data with applications to additive
  manufacturing,'' \emph{IISE Transactions}, vol.~54, no.~5, pp. 464--480,
  2022.

\bibitem{koosha2017statistical}
M.~Koosha, R.~Noorossana, and F.~Megahed, ``Statistical process monitoring via
  image data using wavelets,'' \emph{Quality and Reliability Engineering
  International}, vol.~33, no.~8, pp. 2059--2073, 2017.

\bibitem{Okhrin2021new}
Y.~Okhrin, W.~Schmid, and I.~Semeniuk, ``New approaches for monitoring image
  data,'' \emph{IEEE Transactions on Image Processing}, vol.~30, pp. 921--933,
  2021.

\bibitem{eslami2023spatial}
D.~Eslami, H.~Izadbakhsh, O.~Ahmadi, and M.~Zarinbal, ``Spatial-nonparametric
  regression: an approach for monitoring image data,'' \emph{Communications in
  Statistics-Theory and Methods}, vol.~52, no.~12, pp. 4114--4137, 2023.

\bibitem{sabokrou2015real}
M.~Sabokrou, M.~Fathy, M.~Hoseini, and R.~Klette, ``Real-time anomaly detection
  and localization in crowded scenes,'' in \emph{Proceedings of the IEEE
  conference on computer vision and pattern recognition workshops}, 2015, pp.
  56--62.

\bibitem{sabokrou2017deep}
M.~Sabokrou, M.~Fayyaz, M.~Fathy, and R.~Klette, ``Deep-cascade: Cascading 3d
  deep neural networks for fast anomaly detection and localization in crowded
  scenes,'' \emph{IEEE Transactions on Image Processing}, vol.~26, no.~4, pp.
  1992--2004, 2017.

\bibitem{hurley2009comparing}
N.~Hurley and S.~Rickard, ``Comparing measures of sparsity,'' \emph{IEEE
  Transactions on Information Theory}, vol.~55, no.~10, pp. 4723--4741, 2009.

\end{thebibliography}

\newpage
\quad
\newpage

\appendix

\begin{proof}[Proof of Lemma~\ref{lem:sparsity of white noise}]
First, we transform $1-h(\mathbf e)$ into the following, 
\begin{equation}
    \begin{aligned}
        1-h(\mathbf e) &= \frac{\left|\sum_{i,j}e_{ij}\right|/\|\mathbf e\|_F - 1}{\sqrt{p_1p_2}-1}\\
        &= \left(\sqrt{p_1p_2}-1\right)^{-1}\frac{\left|\sum_{i,j}e_{ij}\right|/\sqrt{p_1p_2}}{\sqrt{\sum_{i,j}e^2_{ij}/p_1p_2}}. 
    \end{aligned}
\end{equation}
By Law of Large Number, we have
\begin{equation}
    \sum_{i,j}e_{ij}^2/p_1p_2 = \operatorname{O}_{a.s.}(1). 
\end{equation}
By Law of the Iterated Logarithm, we have
\begin{equation}
    \sum_{i,j}e_{ij} = \operatorname{O}_{a.s.}\left(\sqrt{p_1p_2\log\log p_1p_2}\right). 
\end{equation}
Plug the above two orders back into $1-h(\mathbf e)$, we arrive at the conclusion. 
\end{proof}

\begin{proof}[Proof of Theorem~\ref{thm:affection of white noise}]
    By the definition of Hoyer-like index, we have the difference 
    \begin{equation}\label{formula:first step hoyer upperbound}
        \begin{aligned}
            &h(\mathbf R_t)-h(\mathbf A) = \left(\frac{\left|\sum_{i,j}A_{ij}\right|}{\|\mathbf A\|_F}-\frac{\left|\sum_{i,j}R_{tij}\right|}{\|\mathbf{R}_t\|_F}\right)\left(\sqrt{p_1p_2}-1\right)^{-1}\\
            & = \left(\frac{\left|\sum_{i,j}A_{ij}\right|}{\|\mathbf A\|_F}-\frac{\left|\sum_{i,j}R_{tij}\right|}{\|\mathbf{A}\|_F}+\right.\\
            &\quad \left.\frac{\left|\sum_{i,j}R_{tij}\right|}{\|\mathbf{A}\|_F}-\frac{\left|\sum_{i,j}R_{tij}\right|}{\|\mathbf{R}_t\|_F}\right)\left(\sqrt{p_1p_2}-1\right)^{-1}\\
            & \leq \left(\frac{\left|\sum_{i,j}d_{ij}+e_{tij}\right|}{\|\mathbf A\|_F}+\right.\\
            & \quad \left.\left|\sum_{i,j}R_{tij}\right|\left|\frac{1}{\|\mathbf{A}\|_F}-\frac{1}{\|\mathbf{R}_t\|_F}\right|\right)\left(\sqrt{p_1p_2}-1\right)^{-1}, 
        \end{aligned}
    \end{equation}
    where $d_{ij} = (\bm \mu_0-\hat{\bm \mu}_0)_{ij}$. 
    We split and conquer the above term. First, we have 
    \begin{equation}
        \begin{aligned}
        &\frac{\left|\sum_{i,j}d_{ij}+e_{tij}\right|}{\|\mathbf A\|_F}\left(\sqrt{p_1p_2}-1\right)^{-1} \\
            & =  \frac{\left|\sum_{i,j}d_{ij}+e_{tij}\right|/p_1p_2}{\|\mathbf A\|_F/\sqrt{p_1p_2}}\frac{\sqrt{p_1p_2}}{\sqrt{p_1p_2}-1} \xrightarrow{a.s.}{} 0. 
        \end{aligned}
    \end{equation}
    Second, we have 
    \begin{equation}
        \begin{aligned}
            &\left|\sum_{i,j}R_{tij}\right|\left|\frac{1}{\|\mathbf{A}\|_F}-\frac{1}{\|\mathbf{R}_t\|_F}\right|\left(\sqrt{p_1p_2}-1\right)^{-1}\\ 
            &= 
            \frac{\left|\sum_{i,j}R_{tij}\right|}{p_1p_2}\left|\frac{1}{\|\mathbf{A}\|_F/\sqrt{p_1p_2}}-\frac{1}{\|\mathbf{R}_t\|_F/\sqrt{p_1p_2}}\right| \frac{\sqrt{p_1p_2}}{\sqrt{p_1p_2}-1}.
        \end{aligned}
    \end{equation}
    Inside the above formula, we have 
    \begin{equation}
        \begin{aligned}
            &\frac{\left|\sum_{i,j}R_{tij}\right|}{p_1p_2}\xrightarrow{\mathbb P}{}\bar{a}, \quad {\|\mathbf{A}\|^2_F/{p_1p_2}} \to {{\bar{a^2}}}, \\[3pt]
            & \|\mathbf{R}_t\|^2_F/{p_1p_2} \\[3pt]
            & = \frac{\sum_{i,j} A_{ij}^2 + d_{ij}^2 + e_{tij}^2 + 2d_{ij}(A_{ij}+e_{tij})+ 2A_{ij} e_{tij}}{p_1p_2}\\
            &\xrightarrow{a.s.}{\bar{a^2}+\sigma^2}, 
        \end{aligned}
    \end{equation}
    where the convergence in quadratic terms and the interaction term involving $d_{ij}$ is straightforward, and the term ${\sum_{i,j} A_{ij} e_{tij}}/{p_1p_2}$ converges to its expectation in $\mathbb L^2$, thus in probability. 
    Plug the limits back into the inequality \eqref{formula:first step hoyer upperbound}, we have 
    \begin{equation}\label{formula:hoyer upperbound}
        \limsup_{{p_1},{p_2}\to\infty} h(\mathbf R_t)-h(\mathbf A) \leq \frac{\bar{a} \sigma^2}{\sqrt{\bar{a^2}\left(\bar{a^2}+\sigma^2\right)}\left(\sqrt{\bar{a^2}}+\sqrt{\bar{a^2}+\sigma^2}\right)}.
    \end{equation}
    On the other hand, consider the lower bound holds as
        \begin{equation}
        \begin{aligned}
            h(\mathbf R_t)-h(\mathbf A) &\geq \left(
            \left|\sum_{i,j}R_{tij}\right|\left(\frac{1}{\|\mathbf{A}\|_F}-\frac{1}{\|\mathbf{R}_n\|_F}\right)\right.\\
            &\quad \left.-\frac{\left|\sum_{i,j}d_{ij}+e_{tij}\right|}{\|\mathbf A\|_F}\right)\left(\sqrt{p_1p_2}-1\right)^{-1}.
        \end{aligned}
    \end{equation}
    Following the same logic, we have
    \begin{equation}\label{formula:hoyer lowerbound}
        \liminf_{{p_1},{p_2}\to\infty}h(\mathbf R_t)-h(\mathbf A)\geq \frac{\bar{a} \sigma^2}{\sqrt{\bar{a^2}\left(\bar{a^2}+\sigma^2\right)}\left(\sqrt{\bar{a^2}}+\sqrt{\bar{a^2}+\sigma^2}\right)}.
    \end{equation}
    The combination of \eqref{formula:hoyer upperbound} and \eqref{formula:hoyer lowerbound} completes the proof.
\end{proof}

\begin{proof}[Proof of Corollary~\ref{coro:divergent noise}]
By the definition of Hoyer index, we have
$$
\lim_{p_1,p_2\to\infty} h(\mathbf{A}) = 1-\frac{\bar a}{\sqrt{\bar{a^2}}}
$$
By Theorem \ref{thm:affection of white noise}, we have
$$
\begin{aligned}
    &\lim_{\sigma^2\to\infty}\lim_{p_1,p_2\to\infty} h(\mathbf{A+e})-h(\mathbf{A}) \\
    &\stackrel{a.s.}{=}\lim_{\sigma^2\to\infty}\frac{\bar{a} \sigma^2}{\sqrt{\bar{a^2}\left(\bar{a^2}+\sigma^2\right)}\left(\sqrt{\bar{a^2}}+\sqrt{\bar{a^2}+\sigma^2}\right)}\\
    &=\frac{\bar a}{\sqrt{\bar{a^2}}}.
\end{aligned}
$$
Combine the above two, we have
$$
\lim_{\sigma^2\to\infty}\lim_{p_1,p_2\to\infty} h(\mathbf{A+e}) \stackrel{a.s.}{=} 1.
$$
\end{proof}
\end{document}